\documentclass[a4paper,twocolumn,11pt,accepted=2021-01-28]{quantumarticle}
\pdfoutput=1
\usepackage[utf8]{inputenc}
\usepackage[T1]{fontenc}
\usepackage{amsmath}
\usepackage{amsfonts}
\usepackage{dsfont} 
\usepackage{graphicx}
\usepackage{todonotes}
\usepackage{gensymb}
\usepackage{bbold}
\usepackage{bm}
\usepackage[numbers,sort&compress]{natbib}
\usepackage{amssymb}
\usepackage{latexsym}
\usepackage{fancyhdr}
\usepackage[bbgreekl]{mathbbol}
\usepackage{epsfig}
\usepackage{color}
\usepackage{textcomp}
\usepackage{amsthm}
\usepackage{comment}
\usepackage{ulem}
\usepackage{hyperref}
\usepackage{mathrsfs}

\newtheorem{theorem}{Theorem}[section]

\newcommand{\ket}[1]{\left| {#1} \right\rangle}
\newcommand{\bra}[1]{\left\langle {#1} \right|}

\newcommand{\ketbra}[2]{\left| {#1} \right\rangle \!\! \left\langle {#2} \right|}

\newcommand{\tr}[1]{\mbox{$\mathrm{Tr}\left(#1\right)$}}

\newcommand{\Id}{\mathds{1}}

\newcommand{\cF}{\mathcal{F}}
\newcommand{\cH}{\mathcal{H}}

\newcommand{\be}{\begin{equation}}
\newcommand{\ee}{\end{equation}}

\newcommand{\ma}{\mathsf{a}}
\newcommand{\mb}{\mathsf{b}}
\newcommand{\mx}{\mathsf{x}}
\newcommand{\my}{\mathsf{y}}

\definecolor{darkgreen}{RGB}{0,121,0}

\newcommand{\stkout}[1]{\ifmmode\text{\sout{\ensuremath{#1}}}\else\sout{#1}\fi}

\begin{document}

\title{Self-testing with finite statistics enabling the certification of a quantum network link}
\author{Jean-Daniel Bancal}
\affiliation{Université Paris-Saclay, CEA, CNRS, Institut de Physique Théorique, 91191, Gif-sur-Yvette, France}
\affiliation{Group of Applied Physics, University of Geneva, 1211 Geneva 4, Switzerland}
\affiliation{Quantum Optics Theory Group, Universität Basel, CH-4056 Basel, Switzerland}
\author{Kai Redeker}
\affiliation{Fakultät für Physik, Ludwig-Maximilians-Universität, 80799 München, Germany}
\author{Pavel Sekatski}
\affiliation{Quantum Optics Theory Group, Universität Basel, CH-4056 Basel, Switzerland}
\affiliation{Group of Applied Physics, University of Geneva, 1211 Geneva 4, Switzerland}
\author{Wenjamin Rosenfeld}
\affiliation{Fakultät für Physik, Ludwig-Maximilians-Universität, 80799 München, Germany}
\affiliation{Max-Planck-Institut für Quantenoptik, Hans-Kopfermann-Strasse 1, 85748 Garching, Germany}
\author{Nicolas Sangouard}
\affiliation{Université Paris-Saclay, CEA, CNRS, Institut de Physique Théorique, 91191, Gif-sur-Yvette, France}
\affiliation{Quantum Optics Theory Group, Universität Basel, CH-4056 Basel, Switzerland}

\date{16 February 2021}
\begin{abstract}
\noindent Self-testing is a method to certify devices from the result of a Bell test. Although examples of noise tolerant self-testing are known, it is not clear how to deal efficiently with a finite number of experimental trials to certify the average quality of a device without assuming that it behaves identically at each run. As a result, existing self-testing results with finite statistics have been limited to guarantee the proper working of a device in just one of all experimental trials, thereby limiting their practical applicability. We here derive a method to certify through self-testing that a device produces states on average close to a Bell state without assumption on the actual state at each run. Thus the method is free of the I.I.D. (independent and identically distributed) assumption. Applying this new analysis on the data from a recent loophole-free Bell experiment, we demonstrate the successful distribution of Bell states over 398 meters with an average fidelity of $\geq$55.50\% at a confidence level of 99\%. Being based on a Bell test free of detection and locality loopholes, our certification is evidently device-independent, that is, it does not rely on trust in the devices or knowledge of how the devices work. This guarantees that our link can be integrated in a quantum network for performing long-distance quantum communications with security guarantees that are independent of the details of the actual implementation.
\end{abstract}

\maketitle
%{Distribution of entanglement over long distances: what for?} 
\section{Introduction}
The distribution of entanglement over long distances is a key challenge in extending the range of quantum communication and building quantum networks~\cite{Kimble08,Gisin07}. The direct transmission of entangled states through optical fibers is a viable solution for short distances but is limited by transmission loss. Quantum repeaters have thus been proposed for entanglement distribution over long distances\footnote{For protocols that do not rely on entanglement creation in independent links, but are based exclusively on error correction codes, see~\cite{Knill96}.}~\cite{Briegel98, Sangouard11}. The basic idea is to divide the global distance into short elementary network links. Entanglement is created in each link and successive entanglement swapping operations are used to combine links and extend entanglement. The key to efficient quantum repeater is to use elementary links where i) the successful creation of entanglement is heralded and ii) entanglement is stored so that it can be created in each link independently. Impressive progress along this line now allows one to envision multipartite quantum networks where entanglement is distributed between arbitrary parties.~\cite{Pirandola16, Schoute16, Pant17, Wallnofer16, Epping16, Das17}\\

%{State identification without assumption on the Hilbert space dimension: a key element} 
At the heart of quantum networks lies the ability to distribute but also certify an entangled state between two distant locations. Although entangled states have been produced between remote locations forming an elementary network link in multiple experiments~\cite{Matsukevich06, Chou07, Yuan08, Moehring07, Matsukevich08, Lee11, Ritter12, Hofmann12, Bernien13}, their suitability for general purposes including -- but not limited to -- quantum key distribution, remains unproven. Demonstrations based on the qubit assumption for instance, stating that all elements involved are of dimension two, are subject to side-channels which completely corrupt the security guarantees~\cite{Acin06}. \\

%{State identification without assumption on the measurement: beyond QKD} 
More generally, the identification of a quantum state provides the most complete description of a system. But the trace left by a state in the measurement outcomes is as much influenced by the state as by the measurement itself. Consequently, it is challenging to obtain an accurate description of a quantum state from observed statistics without presuming a detailed description of the measurement apparatus. Yet, characterizing univocally a quantum resource by identifying its quantum state constitutes a crucial step to set quantum technologies on a solid stand.\\

%{DI state characterization is self-testing, and this provides security for quantum information tasks}
The possibility of device-independent state characterization which is not relying on assumptions on the dimension of the Hilbert space and on the correct calibration or modelling of the measurements~\cite{Scarani12} was first realized in Ref.~\cite{Popescu92, Braunstein92}. There it was noted that the only quantum states able to achieve a maximum violation of the Bell-CHSH inequality~\cite{CHSH69}, are Bell states -- two-qubit maximally entangled states. Interest in self-testing however only started growing significantly after Mayers and Yao rediscovered it and showed that it provides security for quantum key distribution~\cite{Mayers98,Mayers04}. Since then, it has been understood that self-testing guarantees the security of many quantum information tasks, including randomness generation~\cite{Colbeck09,Pironio10,Bamps16} and delegated quantum computing~\cite{Reichard13}; see~\cite{Supic19} for more details. Therefore, self-testing a state guarantees its direct applicability for a wide range of applications. Motivated by this perspective, further theoretical self-testing results have been obtained lately, addressing an increasing range of states, and with improving tolerance to noise~\cite{Tomamichel13, Bardyn09, Miller13, Yang14, Kaniewski16, McKague13, Bamps16, Natarajan17, Coladangelo17, Supic17}. Moreover, self-testing has also been extended to the characterization of quantum measurements and channels~\cite{Mayers04, Magniez06, Reichard13, Bancal15, Dall'Arno17, Kaniewski17, Sekatski18, Wagner19}.\\

%{Experimental status of self-testing (references) and what we add}
In the case of Bell states, it is now known that self-testing based on the Bell-CHSH inequality is strongly resistant to noise~\cite{Bardyn09,Yang14,Kaniewski16}. Recently, this led to the experimental estimation of self-testing fidelities from the perspective of hypothesis testing, in which the null hypothesis to be rejected is that the source only produces states with a fidelity below a fixed threshold value~\cite{Tan17}. Rejection of this hypothesis then implies that at least one state produced by the source had a fidelity higher than the threshold value. However, the implications of hypothesis testing to practical protocols is not clear since no statement on the average fidelity is provided. As an example, data from the experiment involving two individual ions separated by $\approx1$ m~\cite{Pironio10} were shown to lead to a significant rejection of the null hypothesis even though the average Bell violation that could be certified with the methods presented in~\cite{Pironio10} at 99\% confidence level is lower than $2.05$, hence precluding a conclusion on the average Bell state fidelity~\cite{Valcarce19}. A higher Bell violation was demonstrated between two ions separated by about 340 $\mu$m within a trap~\cite{Tan17}, but this short-distance setting is not directly applicable for quantum networks. The recent advent~\cite{Hensen15, Shalm15, Giustina15, Rosenfeld17} of loophole-free Bell tests~\cite{Bell64} involving large separations between entangled particles opens a new perspective for device-independent certification of states distributed in quantum networks.\\

We here derive a method that provides a confidence interval on the \emph{average} violation of any binary Bell inequality without assuming that the trials are independent and identically distributed. Applying this method to the CHSH-Bell test allows us to certify that a source has the capability of producing \emph{on average} states close to a Bell state without making any assumption on the actual state at each trial. This changes the status of self-testing from a mere theoretical tool to a practical certification technique. We show this by considering the data used on the loophole-free Bell inequality violation reported in Ref.~\cite{Rosenfeld17} where entanglement is distributed in a heralded way and stored in two single atoms trapped at two locations separated by 398m before a Bell test is performed. We first optimize the heralding conditions using an ab-initio model of the entanglement generation process, hence improving on the entanglement fidelity of the data set in \cite{Rosenfeld17}. We then apply our new statistical tool accounting for finite experimental statistics and imperfections of the random number generators. From the observed Bell-CHSH value, we certify the successful distribution over 398~m of an entangled state with a Bell state fidelity of 55.50\% at a confidence level of 99\%. This constitutes the first result where a statistically relevant bound on the average fidelity of the distributed state is obtained directly from the Bell-CHSH value and the first device-independent certification of an elementary link for a quantum network.\\

\section{Device-independent assumptions}
The scenario we consider involves three protagonists, colloquially referred to as Alice, Bob and Charlie, see Fig.~\ref{fig:ExpScheme}. Charlie holds a preparation device which indicates when the experiment is ready: it heralds the start of every measurement procedure. The two other parties each hold one measurement device and one random number generator device. Upon heralding, the random number generators are used by Alice and Bob to choose a measurement setting which is applied to their measurement devices. Measurement settings and outcomes are recorded locally for later analysis. The claim of self-testing for the state measured is based on a number of assumptions that we review now.

\begin{enumerate}
\item The experiment admits a quantum description. Essentially, the state of the system can be represented in terms of a density operator, and the measurements as operators acting on the same Hilbert space with the appropriate tensor structure.
\item All devices mentioned above are well identified in space and operate sequentially in time. In particular, the separation between the parties Alice and Bob is clear, as well as between the random number generators and the measurement devices of each party. Moreover, results are recorded before going to the next round, hence we know exactly when a round is going on (two rounds don't happen simultaneously), when it is finished, and we can
monitor how many rounds happened in a given time.
\item The random number generators are independent from all other devices and sample from a well characterized probability distribution. Hence, the measurements used are chosen freely: the measured particles cannot influence this choice, nor vice versa. The random number devices can be correlated to each other, but not to the rest of the setup.
\item Finally, the classical and quantum communication between Alice and Bob is limited: no communication (whether direct or indirect) is allowed between the measurement boxes once the settings choices are received and until the measurement outcomes are produced. Moreover, the random number generators only provide the choice of measurement setting when required, and to their respective measurement device. (Note that space-like separation can be used to guarantee the condition of no communication between Alice and Bob.)
\end{enumerate}

Apart from the first assumption, which has not been challenged by any experiment so far, note that all three remaining assumptions concern the \textit{relation} between the various devices involved in the experiment rather than their internal working. This approach is thus often called ``black-box" or ``device-independent". For a physical setup permitting for self-testing these assumptions are requirements. The settings and requirements for self testing are sufficient to test a Bell inequality and they have been used recently in Ref.~\cite{Rosenfeld17} to perform a loophole-free violation of the Bell-CHSH inequality. We briefly present this experiment in the next section. \\

\begin{figure*}
\centering
\includegraphics{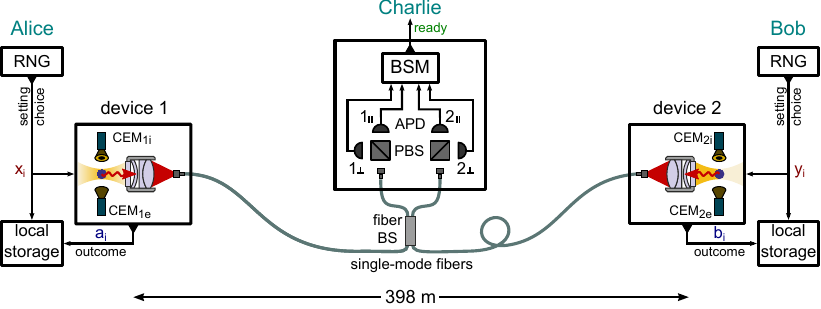}
\caption{Sketch of self-testing based on violation of Bell's inequality with entangled atoms separated by a large distance. Each "device" of Alice and Bob is an independent apparatus for trapping and manipulating single atoms. Entanglement between the atoms is generated by entangling the spin of each atom with polarization of a single photon. The photons are coupled into single mode fibers and overlapped at a fiber beamsplitter. Coincident detection of two photons in Charlie's device heralds the entanglement. Alice and Bob then use their random number generators (RNGs) to select a measurement setting for fast and efficient read-out of the atomic state based on state selective ionization using particle detectors (CEMs) to detect the created ions (i) and electrons (e).}
\label{fig:ExpScheme}
\end{figure*}

\section{Event-ready CHSH-Bell test with neutral atoms}
In our experiment, Alice and Bob's stations are made each with a single $^{87}$Rb atom stored in an optical dipole trap, see Fig.~\ref{fig:ExpScheme}. The two setups are independently operated, that is, they are equipped with their own laser and control systems. Two Zeeman states $|m_F=\pm1\rangle$ of the ground state manifold $5^2S_{1/2}$ are used as 1/2-spin states. After an initial state preparation, the atoms are optically excited to emit a photon whose polarization is entangled with the atomic spin states, see Fig.~\ref{fig:Levels}(a). The photons are sent to Charlie's station that is located close to Alice's location where a Bell state measurement is implemented with a beamsplitter followed by a polarizing beamsplitter at each output port and four single photon detectors, see Fig.~\ref{fig:ExpScheme}. The atom excitation procedure is synchronized on a timescale that is much shorter than the photon duration. Careful adjustment of experimental parameters ensures a spectral, temporal and spatial mode overlap of photons close to unity \cite{Hofmann12}. This allows us to achieve a high two-photon interference quality limited mostly by two-photon emission effects of a single atom. The joint measurement performed on these photons distinguishes two out of the four Bell states and ideally projects the atoms into either of the two states $|\psi^\pm\rangle = (\ket{\uparrow}_x\ket{\downarrow}_x\pm\ket{\downarrow}_x\ket{\uparrow}_x)/\sqrt{2}$ according to the outcome. Depending on the loading rate of the traps, 1 to 2 successful Bell state measurements are obtained per minute. At each success, a signal is sent to Alice and Bob and triggers setting choices. The analysis basis is selected by the output of a fast quantum random number generator, which is based on counting photons emitted by an LED with a photo-multiplier tube \cite{Fuerst2010,Rosenfeld17}. The measurement outcome is obtained by a spin-state dependent ionization with a fidelity of $97\%$ on a timescale  $\leq 1.1\mu$s. Given that Alice's and Bob's locations are separated by  398 m, this warrants space-like separation of the measurements. Although (strict) space-like separation is not a necessary condition for self-testing, it is a strong guarantee that Alice and Bob's measurement devices are indeed separated from each other and that information about the setting of one party is not available to the other one upon measurement, i.e.~for assumptions 2 and 4 above. \\

\section{CHSH Bell inequality}
Let us label the measurement settings $x=0,1$ and $y=0,1$ for Alice and Bob respectively, with  outcomes $a=0,1$ and $b=0,1$ for each spin measurement. For each pair of settings, we define the correlator $E_{xy}=\sum_{ab}(-1)^{a+b}P(a,b|x,y)$ where $P(a,b|x,y)$ is the conditional probability of observing outcome $a$ and $b$ when choosing the settings $x$ and $y$. This allows us to define the Bell-CHSH value, given by 
\begin{equation}\label{eq:CHSH}
S = E_{00} + E_{01} + E_{10} - E_{11}.
\end{equation}
The latter is upper bounded by $2$ for any local causal theory~\cite{CHSH69}. A significant violation of this bound can thus rule out this possibility, as conclusively demonstrated earlier~\cite{Rosenfeld17}, see also~\cite{Hensen15, Shalm15, Giustina15}. Note that here the values of 0 and 1 for the settings and outcomes were assigned arbitrarily, therefore, any of the 8 relabellings of Eq.~\eqref{eq:CHSH} equally qualifies as a valid definition of the quantity $S$~\cite{Rosset14}. With fixed measurement settings, such equivalent rewritings of the CHSH expression may be necessary to obtain a violation of the local bound with different Bell states.\\

\section{Self-testing a Bell state}
Given assumption 1 above, we can associate to each measurement of Alice and Bob quantum observables $A_x$ and $B_y$ acting on two Hilbert spaces $\cH_A$ and $\cH_B$ of unknown dimension. Also, we can define the quantum state shared by the two parties as $\rho_{AB} \in L(\cH_A\otimes\cH_B).$ We emphasize that the internal functioning of the source and measurement boxes do not need to be known. We simply attribute a quantum state and measurement operators to the actual implementation. Our aim is to identify the actual state $\rho_{AB}$ from the observed statistics only. More precisely, we wish to estimate its fidelity with respect to a maximally entangled state of two qubits, that is
\begin{equation}\label{eq:BellStateFidelity}
F(\rho_{AB}) = \underset{\Lambda_A,\Lambda_B}{\max}\tr{(\Lambda_A \otimes \Lambda_B)[\rho_{AB}], \ketbra{\psi^-}{\psi^-}},
\end{equation}
where the maximization is over all local trace-preserving maps $\Lambda_{A/B}:\cH_{A/B} \to \mathbb{C}^2.$ The role of these maps $\Lambda_{A/B}$ is to identify the subsystems inside the unknown Hilbert spaces $\cH_{A/B}$ in which $\rho_{AB}$ can be compared to the desired state. Given an observed Bell-CHSH value $S$, the Bell state self-testing fidelity is defined as the minimum fidelity of the unknown quantum state $\rho_{AB}$ which is compatible with the violation, i.e.
\begin{eqnarray}
\cF = \underset{\rho_{AB}, {A_x}, {B_y}}{\min}&& F(\rho_{AB})\\ 
\nonumber \text{s.t.\ \ \ }&& E_{00} + E_{01} + E_{10} - E_{11} = S,
\end{eqnarray}
where the correlators are now given by $E_{xy}=\tr{\rho_{AB} A_xB_y}$.
This quantity captures the relation between $\rho_{AB}$ and the singlet state $\ket{\psi^-}$, one representative Bell state, that can be inferred from observed statistics: if the quantum state is separable, then $\cF\leq\frac12$; on the other hand, if $\cF=1$, then we have the guarantee that local maps exist which identify perfectly a Bell state within the state $\rho_{AB}$, because this is the case for all admissible quantum realizations.\\

It has been shown that the self-testing fidelity $\cF$ can be directly related to the sole knowledge of the Bell-CHSH value $S$~\cite{Bardyn09}. The tightest known relation is given by~\cite{Kaniewski16}
\begin{equation}\label{eq:fidCHSH}
\cF \geq f(S)=\frac{\max\left(40,\,12+(4+5\sqrt{2})(5S-8)\right)}{80}.
\end{equation}

\section{Statistical analysis}
The previous formula holds in the limit where the CHSH value $S$ is known exactly. In order to analyze a real experiment with finite statistics, we consider that each run $i=1,\ldots,n$ is characterized by an (unknown) CHSH value $S_i$ and fidelity $\cF_i$. This fidelity could be different at each round, and depend on past events. We are then interested in making a claim on the average fidelity
\begin{equation}\label{eq:averageFidelity}
\overline{\cF} = \frac{1}{n}\sum_{i=1}^n \cF_i.
\end{equation}
Other works have considered a different figure of merit for the certification of states in a non-I.I.D. setting~\cite{Supic19}. In Appendix~\ref{app:Fidelities}, we show that the two approaches are equivalent to each other and that the numerical value provided by the average fidelity~\eqref{eq:averageFidelity} has the advantage of being a direct quantifier of the source quality.

Assuming that the measurement settings are chosen independently by both parties and with a maximum bias $\tau$ with respect to a uniform distribution, i.e.~$1/2-\tau \leq P(x),P(y) \leq 1/2+\tau$, we show in the Appendix~\ref{App:statistics} that
\begin{equation}\label{eq:S-hat}
\hat S = 8\left(I^{-1}_{\alpha}\!\left(n \overline t-1, n(1-\overline t)+2\right)-\tau-\tau^2\right)-4
\end{equation}
is a lower bound on the average CHSH violation $\overline S=\frac{1}{n}\sum_i S_i$ with confidence level $1-\alpha$. This allows us to conclude that $[\hat{\cF}, 1]$ with
\begin{equation}\label{eq:cf-hat}
\hat \cF = f(\hat S)
\end{equation}
is a one-sided confidence interval for $\overline{\cF}$ with confidence level $1-\alpha$. Here, $\overline t = (4+\overline S^u)/8$ with $\overline S^u$ is the average CHSH value observed over the $n$ rounds assuming a uniform sampling of the settings, i.e.~following Eq.~\eqref{eq:Su} from the Appendix~\ref{App:statistics}, and $I^{-1}$ is the inverse regularized incomplete Beta function, i.e.~$I_\my(\ma,\mb)=\mx$ for $\my=I^{-1}_\mx(\ma,\mb)$. We emphasize that this bound on the average Bell state fidelity does not rely on the I.I.D. assumption.\\

\section{Preselection}
In contrast to results presented in Ref.~\cite{Rosenfeld17}, where all registered events were taken into account, here we use a pre-selected set of events to compute the Bell-CHSH violation and the subsequent self-testing fidelities of heralded atomic states. This selection is based on a physical model which takes into account detrimental two-photon emission effects of a single atom and allows to define pre-selection criteria, here a time-window for acceptance of photons in the BSM, to improve the fidelity of the entangled atom-atom state. Details can be found in the Appendix~\ref{app:Preselection}.
Importantly, these considerations are not based on the results observed during the experiment. They are based on an ab-initio model of the underlying excitation and emission processes. Therefore, these considerations allow for the determination of a significance level and desired amount of data prior to the data acquisition stage, in agreement with the requirements of confidence intervals construction. This selection can then be seen as a pre-selection of the data, or equivalently, as a state preparation. In particular, it does not open the detection loophole or introduce expectation bias~\cite{Jenga06}.\\

\begin{figure}
\centering
\includegraphics[width=8.6cm]{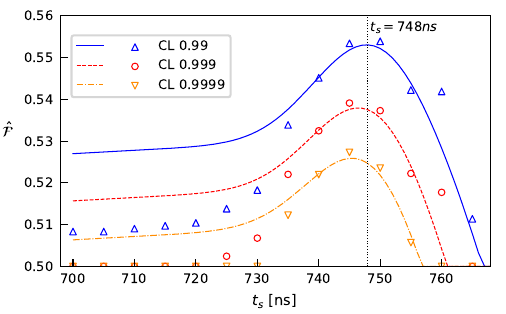}
\caption{Expected self-testing fidelity $\hat{\cF}$ resulting from the pre-selection model (lines) for different confidence levels (CL) as a function of the starting time of the acceptance time window for heralding events $t_{s}$ (at a fixed $t_{e}$). The optimal start times for for different confidence levels according to the model are shown in Tab.~\ref{tab:Fidelities2}. In comparison the self-testing fidelity for the measured data ($|\psi^-\rangle$, symbols) evaluated with the acceptance time window starting with $t_s$. Note that the model for the pre-selection is ab-initio and not a fit for the data. For details of the model and a comparison to the measured data see Appendix.}
\label{fig:self-testing-fidelity}
\end{figure}

\section{Results}
For the evaluation we use the data of events heralding the $\ensuremath{|\psi^{-}\rangle}$ state from the loophole-free Bell test~\cite{Rosenfeld17} taken between 05.02.2016 and 24.06.2016 ($25211$ events). Fig.~\ref{fig:self-testing-fidelity} shows the resulting lower bound of the average fidelity $\hat{\mathcal{F}}$ for the ab-initio model and the data set using the same pre-selection as a function of the acceptance window starting time $t_{s}$ for different confidence levels. The model allows to determine the acceptance time window start time $t_s$ and end time $t_e$ for an optimal expected lower bound for the fidelity $\hat{\mathcal{F}}$ for each confidence level shown in Fig.~\ref{fig:self-testing-fidelity}. The results for the pre-selected data are shown in Tab.~\ref{tab:Fidelities2}. For calculation of the lower bound of the fidelity $\hat{\mathcal{F}}$ we consider bias of the RNGs bounded by $\tau=6.3\times10^{-4}$ (arsing from the ``paranoid'' model for the predictability~\cite{Rosenfeld17}).

The lower bound of fidelity exceeding the value of $0.5$ (at a confidence level of up to  $1-1.0\times10^{-7}$ for $t_s=746$ ns) represents the first device-independent certification of a distributed entangled state. Moreover, an evaluation of the full data set without any pre-selection yields a Bell state fidelity of $0.5061$ at $99\%$ confidence and a Bell state fidelity larger than $0.5$ can be certified even at a significance level as high as $99.7\%$.

As a comparison, note that a confidence interval could also be constructed from Hoeffding's inequality~\cite{Hoeffding63}, yielding
\begin{equation}\label{eq:hoeff}
\hat S_H = 8\left(\overline t - \sqrt{\frac{\log(1/\alpha)}{2n}}-\tau-\tau^2\right)-4.
\end{equation}
The conclusion obtained with this inequality would however be significantly weaker\footnote{Perhaps counter-intuitively, the improved tail bound given in Eq.~(2.1) of~\cite{Hoeffding63} is not of much help here since it relies on the knowledge of the average winning probability, but confidence intervals must hold for all (unknown a priori) winning probabilities.}. For instance, the claim that the fidelity $\overline \cF$ is nontrivial (i.e.~$\geq 1/2$) for the whole data set would not be statistically significant. Indeed, the corresponding statistical level is $\alpha=6.5\%$, i.e.~about 20 times larger than guaranteed by our bound. The average fidelity guaranteed with pre-selection would also be significantly lower. Using $t_s=748$~ns, for instance, the lower bound on the fidelity at $99\%$ confidence level is $0.5291$, i.e.~about twice closer to the trivial value of $1/2$ compared to $0.5550$.

Additionally, we applied our method to the data of the table-top Bell test performed between two ions separated by 1 meter reported in~\cite{Pironio10}. Our statistical analysis yields an average Bell violation of $2.2715$ or higher at $99\%$ confidence level, clearly above the threshold $2.11$. This sets a lower bound on the Bell state fidelity of $61.46\%$ at this confidence level, hence guaranteeing for the first time that the states distributed in this experiment had a significant average Bell state fidelity.

Finally, we applied our method to data from~\cite{Hensen16} obtained at a distance of $1.3$~km. Due to the limited number of events ($545$), the method can only guarantee a fidelity larger than $50\%$ with a confidence of $\sim94.2\%$. Still, this demonstrates that the method can be used in different systems without the need of knowing their specific details.\\

\begin{table}
\begin{centering}
\begin{tabular}{|c|c|c|c|c|}
\hline 
CL & $\hat{\mathcal{F}}$ &  $t_s$ & $n$ & $\overline{S}^u$ \tabularnewline
\hline 
\hline 
$99\%$ & $0.5550$ & $748$ ns & $13141$ & $2.2589$ \tabularnewline
\hline 
$99.9\%$ & $0.5407$ & $746$ ns & $14807$ & $2.2554$ \tabularnewline
\hline 
$99.99\%$ & $0.5272$ & $745$ ns & $15671$ & $2.2505$ \tabularnewline
\hline 
\end{tabular}
\par\end{centering}
\caption{Fidelity $\hat{\mathcal{F}}$ at different confidence levels. For each the data is pre-selected with an optimal start time for the acceptance time window [$t_s$, $t_e=895\text{\,ns}$] resulting in the corresponding $n$ events and an average CHSH value $\overline{S}^u$.}
\label{tab:Fidelities2}
\end{table}

\section{Discussion}
We have derived a bound on the average fidelity of a measured state with respect to a Bell state from the sole knowledge of the observed Bell-CHSH value which is free of the I.I.D. assumption. This bound was achieved by constructing a non-I.I.D. confidence interval for the sum of $n$ independent binary random variables. We used it to quantify device-independently the quality of a bipartite state distributed over $398$~m in a real-world elementary quantum network link. These results guarantee that this link is suitable for an integration in a quantum network, either directly or as a part of a quantum repeater.\\

\section*{Acknowledgments} J.-D.B., P.S. and N.S. acknowledge funding by the Swiss National Science Foundation (SNSF), through the Grant PP00P2-179109 and 200021-175527, by the Army Research Laboratory Center for Distributed Quantum Information via the project SciNet and from the European Union’s Horizon 2020 research and innovation programme under grant agreement No 820445 and project name Quantum Internet Alliance. K.R. and W.R. acknowledge funding by the German Federal Ministry of Education and Research via the project Q.com-Q.

\appendix
%\newpage

\section{Preselection of heralding events}\label{app:Preselection}

To allow filtering as a preselection we have developed an ab-initio physical model independently of our measurement results to find an optimum filtering based on the model only. This model describes the photon emission process of a single atom excited by a short laser pulse and takes into account all important processes within its multilevel structure. Thereby we are able to calculate the expected fidelity for the entangled state of two atoms heralded by a two-photon coincidence at a certain time. The full description of the model goes far beyond the focus of the present work and can be found in~\cite{thesisJulian14,thesisKai20}, here we only present a brief sketch of it.

\begin{figure}
\begin{centering}
\includegraphics{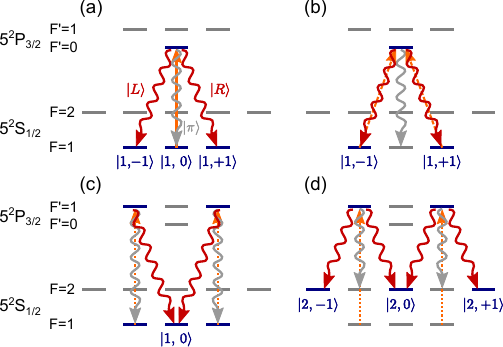}
\par\end{centering}
\caption{Structure of the relevant levels in $\mathrm{^{87}Rb}$, excitation an decay processes. (a) Generation of atom-photon entanglement in spontaneous decay of the excited $5^{2}P_{3/2},F'=0,m_{F}=0$ level. The excitation laser is shown in orange. Photons polarized linearly along the quantization axis ($\pi$-decays, gray) are not detected in our system. After the first decay, a second excitation is possible due to polarization misalignment (b), or off-resonant excitation (c). If the $5^{2}P_{3/2},F'=1$
is excited, also decay to $5^{2}S_{1/2},F=2$ level is possible (d).
\label{fig:Levels} }
\end{figure}

For generation of a photon whose polarization is entangled with the atomic spin state, the atom is excited by a laser pulse resonant to the transition $5^{2}S_{1/2},F=1\rightarrow5^{2}P_{3/2},F'=0$. The temporal shape of the pulse is approximately Gaussian with a FWHM of $22$ ns, see Fig.~\ref{fig:PhotonHisto}. After the successful emission of a photon, ideally, the atom should not interact with the excitation laser due to selection rules, see Fig.~\ref{fig:Levels}(a). In practice, however, the atomic state remains weakly sensitive to the excitation laser due to two reasons. First, there unavoidably are small polarization misalignments of the excitation laser, i.e.~its polarization is not perfectly aligned along the quantization axis (imperfect $\pi$-polarization), allowing for a reexcitation of the
$5^{2}P_{3/2},F'=0,m_{F}=0$ level, see Fig.~\ref{fig:Levels}(b).
Second, off-resonant scattering via the $5^{2}P_{3/2},F'=1$ level is possible (Fig.~\ref{fig:Levels}(c),(d)). Moreover, before the emission of a photon into the desired mode takes place, there is a finite probability that the atom emitted a first photon in a $\pi$-transition, which is not collected by the optics. These multiple photon emissions are detrimental for the quality of the atomic state announced by detection of the photons in the Bell state measurement in two different ways. On the one hand, the state of the atom can be changed by scattering additional photons. On the other hand, the interference quality of photons is reduced since the temporal shape and coherence of the photonic wavepackets is affected.

\begin{figure}
\begin{centering}
\includegraphics{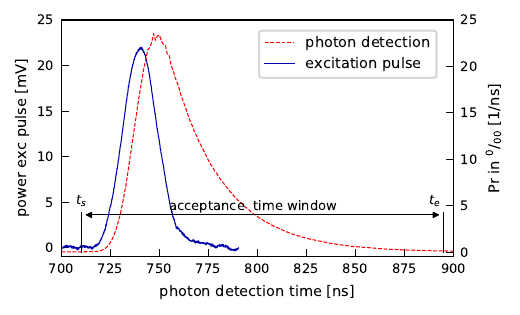}
\par\end{centering}
\caption{Time histogram of single photon detection (red) for excitation of a single atom by a short resonant laser pulse (blue). $t_{s}$ and $t_{e}$ define the acceptance time window for coincidences. \label{fig:PhotonHisto}}

Importantly, the unwanted multiphoton processes happen predominantly during the excitation. Thus, to filter them, we only accept the detection events for which the first detector click is obtained after a time $t_{s}$ when the excitation laser pulse is essentially off, see Fig.~\ref{fig:PhotonHisto}. Additionally, to reduce the dark counts contribution to the heralding event we define a maximal time $t_{e}$ for the detection of the second photon. A later start time $t_s$ increases the entanglement swapping fidelity and by this the expected S-value for the CHSH inequality but  at the expense of the obtained events number, see Fig.~\ref{fig:S-window}. Since the measured S-values will depend not only on the entanglement swapping fidelity but also on other properties,  e.g., the atomic state measurement fidelity and the coherence time of the entangled states, we use the experimental parameters as specified in~\cite{Rosenfeld17} to predict the experiment's S-value.

The optimal selection of the time window of  $\left[t_{s}=748\text{\,ns},t_{e}=895\text{\,ns}\right]$, considering Eq.~\eqref{eq:cf-hat} for a $99\%$ confidence interval, reduces the number of heralding events by approximately a factor of $2$ but the atoms are expected to be in an entangled state of a higher quality.

\end{figure}
\begin{figure}
\begin{centering}
\includegraphics{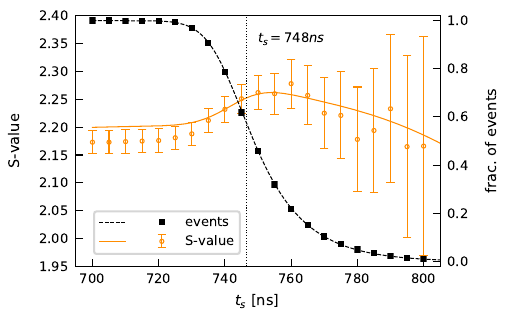}
\par\end{centering}
\caption{S-value for CHSH inequality (orange line) resulting from the ab-initio model as a function of the start time $t_{s}$ of the acceptance time window. The orange circles show the values obtained from experimental data by applying the same filtering criteria. The black dashed line and squares respectively show the predicted and measured fraction of remaining events after filtering. The optimal calculated start time for the acceptance time window is $t_{s}=748\text{\,ns}$ according to Eq.~\eqref{eq:cf-hat} for $99\%$ confidence. The contribution of $t_{e}$ to the fidelity was found to be negligible, it was fixed at $t_{e}=890\text{\,ns}$.
\label{fig:S-window}}
\end{figure}

\section{Finite statistics analysis}\label{App:statistics}
In this section, we detail the construction of the confidence interval on the average singlet fidelity reported in the main text.

\subsection{Model}
In the experimental situation described in the main text, the settings used after the $i^\text{th}$ heralding event can be described by two random variables $X_i$ and $Y_i$. These variables follow a global probability distribution
\begin{equation}\label{eq:settings}
P(\vec X = \vec x, \vec Y = \vec y),
\end{equation}
where $\vec X = (X_1, X_2, \ldots)$, $\vec Y = (Y_1, Y_2, \ldots)$ and $\vec x, \vec y \in\{0,1\}^n$ for a binary choice of settings.

Similarly, two random variables $A_i$ and $B_i$ can be used to describe the outcomes observed upon measuring the state in the $i^\text{th}$ round. By assumptions 2-4 of the main text, the settings and outcomes follow a joint probability distribution of the form
\begin{equation}\label{eq:model}
\begin{split}
P(\vec A &= \vec a, \vec B = \vec b, \vec X = \vec x, \vec Y = \vec y) =\\
& P(\vec X=\vec x, \vec Y=\vec y) \prod_{i=1}^n P_i(a_i,b_i|x_i,y_i,\text{past}_i).\\
\end{split}
\end{equation}
Here, $\vec a, \vec b \in \{0,1\}^n$ are the possible outcome strings in the binary case,
\begin{equation}\label{eq:distri}
\begin{split}
&P_i(a,b|x,y,\text{past}_i)=\\
&\ \ P(A_i=a,B_i=b|X_i=x,Y_i=y,\text{Past}_i=\text{past}_i)
\end{split}
\end{equation}
describes the behavior sampled in the $i^\text{th}$ round and $\text{past}_{i}\ni\{a_j,b_j,x_j,y_j\}_{j<i}$ stands for any information available from the past of round $i$ such as the previous settings and outcomes. The settings distribution can be decomposed into the measurement rounds in a similar fashion as
\begin{equation}
P(\vec X = \vec x, \vec Y = \vec y) = \prod_{i=1}^n P_i(x_i,y_i|\text{past}_i) 
\end{equation}
with
\begin{equation}
P_i(x,y|\text{past}_i)=P(X_i=x,Y_i=y|\text{Past}_i=\text{past}_i).
\end{equation}

We associate to each measurement round $i$ the CHSH value
\begin{equation}\label{eq:CHSH def}
S_{i|\text{past}_i} = \sum_{a,b,x,y} (-1)^{a+b+x y} P_i(a,b|x,y,\text{past}_i).
\end{equation}
This quantity $S_{i|\text{past}_i}$ is the expectation value for the S parameter given in Eq.~\eqref{eq:CHSH} for the given round $i$. We also define the singlet fidelity $\cF_{i|\text{past}_i}$, which is bounded according to Eq.~\eqref{eq:fidCHSH} as
\begin{equation}
\cF_{i|\text{past}_i} \geq f(S_{i|\text{past}_i}).
\end{equation}
Note that these statistical parameters may be different for all rounds $i$ and may depend on past events.

\subsection{Estimation}
Before focusing on the fidelity, our figure of merit, let us estimate the Bell contribution corresponding to a given round $i$. For this, we introduce the statistic
\begin{equation}\label{eq:Ti}
T_{i|\text{past}_i} = \frac{1}{4}\cdot\frac{\chi(A_i\oplus B_i=X_i Y_i)}{P_i(X_i,Y_i|\text{past}_i)},
\end{equation}
where $\chi$ is the indicator function, i.e.~$\chi(\text{condition})=1$ if the condition is true and $\chi(\text{condition})=0$ for a false condition, $\oplus$ is the addition modulo 2, and the term in the denominator
\begin{equation}
\begin{split}
P_i(X_i,&Y_i|\text{past}_i)\\
=&\begin{cases}
P_i(0,0|\text{past}_i), & \text{if } X_i=0,Y_i=0\\
P_i(0,1|\text{past}_i), & \text{if } X_i=0,Y_i=1\\
P_i(1,0|\text{past}_i), & \text{if } X_i=1,Y_i=0\\
P_i(1,1|\text{past}_i), & \text{if } X_i=1,Y_i=1\\
\end{cases}
\end{split}
\end{equation}
refers to the probability with which the observed settings have been sampled, a notation customary in statistics (c.f.~the usual definition of Fisher information for instance). The expectation value of this estimator is directly related to the CHSH violation on round $i$ given the past:
\begin{equation}\label{eq:TiS}
\begin{split}
&\mathbb{E}(T_{i|\text{past}_i})\\
&\ = \sum_{\vec a, \vec b, \vec x, \vec y} T_{i|\text{past}_i} P(\vec A=\vec a, \vec B=\vec b, \vec X=\vec x,\vec Y=\vec y)\\
&\ = \frac{1}{4}\sum_{a,b,x,y}\frac{\chi(a\oplus b=x y)}{P(X_i=x,Y_i=y|\text{Past}_i=\text{past}_i)}\\
&\ \ \times P(A_i=a,B_i=b,X_i=x,Y_i=y|\text{Past}_i=\text{past}_i)\\
&\ = \frac{1}{4}\sum_{a,b,x,y}\chi(a\oplus b=x y)\ P_i(a,b|x,y,\text{past}_i)\\
&\ = \frac{4+S_{i|\text{past}_i}}{8}.
\end{split}
\end{equation}
This expression thus provides a good estimation of the Bell violation contribution of round $i$. Note that the relation~\eqref{eq:TiS} is valid for all distribution of the settings which is independent from $A$ and $B$'s behavior according to Eq.~\eqref{eq:model}.

In the case where the settings are chosen uniformly, i.e.~$P_i(x_i,y_i|\text{past}_i)=\frac{1}{4}$, the random variable $T_{i|\text{past}_i}$ is a Bernoulli variable whose only possible values are $0$ and $1$. It can then be interpreted as a binary game which is either won (if $T_{i|\text{past}_i}=1$) or lost (if $T_{i|\text{past}_i}=0$). The CHSH contribution of round $i$ can then be re-interpreted in terms of the winning probability $q_{i|\text{past}_i}=P(T_{i|\text{past}_i}=1)$ of this game, such that
\begin{equation}\label{eq:Sq}
8q_{i|\text{past}_i}=4+S_{i|\text{past}_i}.
\end{equation}

\subsection{Settings choice bias}
In practice, it may be difficult to guarantee that the choice of settings is exactly uniform. One can then resort to a partial characterization of the settings' distribution. For instance,  consider the case where the settings of Alice and Bob are chosen independently as
\begin{equation}
P_i(x_i,y_i|\text{past}_i)= P(x_i|\tau^x_i)P(y_i|\tau^y_i)
\end{equation}
with
\begin{align}
P(x_i|\tau^x_i) &= \frac{1}{2}+(-1)^{x_i}\tau^x_i \\
P(y_i|\tau^y_i) &= \frac{1}{2}+(-1)^{y_i}\tau^y_i,
\end{align}
and we only have the guarantee that the local biases are bounded $|\tau_i^x|, |\tau_i^y|\leq \tau$ by some maximal value $\tau\leq \frac{1}{2}$. In this case the statistic~\eqref{eq:Ti} as well as the CHSH value Eq.~\eqref{eq:CHSH def} cannot be evaluated directly. We can nevertheless bound its behavior.

For this, let us then consider the statistic that would correspond to a uniform choice of settings
\begin{equation}\label{eq:Tiu}
\begin{split}
T^u_{i|\text{past}_i} = \chi(A_i \oplus B_i=X_i Y_i).
\end{split}
\end{equation}
As mentioned before, this statistic is a Bernoulli random variable taking value either 0 or 1. Its winning probability is $q_{i|\text{past}_i}^u=\mathbb{E}(T^u_{i|\text{past}_i})$, and can be evaluated without the knowledge of the settings distribution. 
It's expectation value is given by
\begin{align}
&\mathbb{E}(T^u_{i|\text{past}_i})\nonumber\\
& =\sum_{a,b,x,y}\!\!\!\!\!\chi(a \oplus b=x y) P_i(a,b|x,y,\text{past}_i)\\
&\ \ \ \ \ \ \ \ \ \times P(x|\tau^x) P(y|\tau^y) \nonumber \\
&= \sum_{a,b,x,y}\!\!\chi(a \oplus b=x y) P_i(a,b|x,y,\text{past}_i) \nonumber\\
&\ \ \ \ \ \times \left(\frac{1}{4}+\frac{(-1)^{x}}{2}\tau^x + \frac{(-1)^{y}}{2}\tau^y + (-1)^{x+y}\tau^x\tau^y\right) \nonumber\\
&= \underbrace{\frac{1}{4} \sum_{a,b,x,y}\!\!\chi(a \oplus b=x y) P_i(a,b|x,y,\text{past}_i) }_{=\mathbb{E}(T_{i|\text{past}_i}) }\\
&\ \ \ +\sum_{a,b,x,y}\!\!\chi(a \oplus b=x y) P(a,b|x,y,\text{past}_i)\nonumber\\
&\ \ \ \times \left(\frac{(-1)^{x}}{2}\tau^x + \frac{(-1)^{y}}{2}\tau^y + (-1)^{x+y}\tau^x\tau^y\right) \nonumber
\end{align}
Defining $f_{xy}= \sum_{a,b} \chi(a \oplus b = x y) P(a,b|x,y,\text{past}_i)\in[0,1]$ we write
\begin{align}\label{eq: means diff}
&\mathbb{E}(T^u_{i|\text{past}_i}) -\mathbb{E}(T_{i|\text{past}_i}) \\
&= \sum_{x,y} f_{xy} \left(\frac{(-1)^{x}}{2}\tau^x + \frac{(-1)^{y}}{2}\tau^y + (-1)^{x+y}\tau^x\tau^y\right).\nonumber
\end{align}
Let us now consider this sum. Without loss of generality we set $0\leq \tau_y\leq\tau_x\leq \tau$, all the other cases directly follow by a permutation of the outcomes or the exchange of $x$ and $y$. One has 
\begin{align}
& 2 \sum_{x,y} f_{xy}
\left(\frac{(-1)^{x}}{2}\tau^x + \frac{(-1)^{y}}{2}\tau^y + (-1)^{x+y}\tau^x\tau^y\right) \nonumber\\
&= f_{00} \left(\tau^x + \tau^y + 2 \tau^x\tau^y\right) + f_{10} \left(-\tau^x+ \tau^y - 2 \tau^x\tau^y\right) \nonumber\\
&\ \ \ + f_{01} \left(\tau^x - \tau^y- 2 \tau^x\tau^y\right) + f_{11} \left(-\tau^x - \tau^y + 2 \tau^x\tau^y\right) \nonumber\\
&\leq \tau^x + \tau^y + 2 \tau^x\tau^y + f_{01}\left(\tau^x - \tau^y- 2 \tau^x\tau^y\right),
\end{align}
where we used $f_{xy}\in[0,1]$, $\tau^x + \tau^y + 2 \tau^x\tau^y \geq 0$,  $-\tau^x + \tau^y- 2 \tau^x\tau^y \leq 0$ and $-\tau^x - \tau^y + 2 \tau^x\tau^y  \leq 0$. 
For the last term  one finds
\be
\begin{split}
f_{01}&\left(\tau^x - \tau^y- 2 \tau^x\tau^y\right) \\
&\ \ \ \leq \begin{cases}
\tau^x - \tau^y- 2 \tau^x\tau^y & \tau_y < \frac{\tau_x}{1+2\tau_x}\\
0 & \tau_y\geq \frac{\tau_x}{1+2\tau_x},
\end{cases} 
\end{split}
\ee
leading to 
\begin{align}
& \sum_{x,y}f_{xy}\left(\frac{(-1)^{x}}{2}\tau^x + \frac{(-1)^{y}}{2}\tau^y + (-1)^{x+y}\tau^x\tau^y\right)\nonumber\\
&\leq \begin{cases}
 \tau^x  & \tau_y <  \frac{\tau_x}{1+2\tau_x}\\
 \frac{1}{2}(\tau^x + \tau^y) + \tau^x\tau^y & \tau_y\geq  \frac{\tau_x}{1+2\tau_x}
\end{cases}\\
&\leq \tau +\tau^2,
\end{align}
which holds for all allowed values of $\tau^x$ and $\tau^y$. Plugging this inequality in Eq.~\eqref{eq: means diff} then gives
\begin{equation}
\mathbb{E}(T^u_{i|\text{past}_i}) -\mathbb{E}(T_{i|\text{past}_i}) \leq \tau+\tau^2.
\end{equation}
Therefore, we obtain
\be\label{eq:quq}
q_{i|\text{past}_i} \geq q_{i|\text{past}_i}^u  -\tau -\tau^2,
\ee
meaning that a lower bound on the winning probability $q_{i|\text{past}_i}^u$ of the uniform statistic $T^u_{i|\text{past}_i}$ gives rise to a lower bound on $q_{i|\text{past}_i}$. In order to estimate $q_{i|\text{past}_i}$ with a distribution of settings which is not fully known, we can thus safely estimate the CHSH value with the statistic~\eqref{eq:Tiu}, effectively assuming that the settings are chosen uniformly, and then correct the winning probability $q_{i|\text{past}_i}^u$ according to the value of $\tau$, as expressed in~\eqref{eq:quq}. This provides a lower bound on the actual winning probability $q_{i|\text{past}_i}$ of~\eqref{eq:Ti}.

To simplify the notation, we now drop the explicit conditioning on the past, and thus simply write e.g.~$S_i$, $q_i$, $T_i$ and $\cF_i$ for the quantities introduced above.

\subsection{Bounding the fidelity}
We construct a statistical parameter for the whole experiment corresponding to the average Bell state fidelity:
\begin{equation}
\overline{\cF} = \frac{1}{n}\sum_i \cF_i.
\end{equation}
Thanks to the convex relation between the CHSH violation and the singlet fidelity Eq.~\eqref{eq:fidCHSH}, this quantity can be bounded from the average CHSH violation $\overline S$ as
\begin{eqnarray}
\overline{\cF} &=& \frac{1}{n}\sum_i \cF_i\\
&\geq& \frac{1}{n}\sum_i f(S_i)\\
&\geq& f\left(\frac{1}{n}\sum_i S_i\right) =  f(\overline S),
\end{eqnarray}
or equivalently, using relation~\eqref{eq:Sq}, from the average winning probability $\overline{q} = \frac{1}{n}\sum_i q_i$ as
\begin{equation}
\overline{\cF} \geq f(8\overline{q}-4).
\end{equation}
In particular, a left-sided confidence interval for $\overline{q}$ gives rise to a left-sided confidence interval for the singlet fidelity. By relation~\eqref{eq:quq}, a left-sided confidence interval for $\overline{q}^u=\sum_i q_i^u$ gives rise to a left-sided confidence interval for $\overline q$, and thus also for the singlet fidelity:
\begin{equation}
\overline{\cF} \geq f\left(8(\overline{q}^u -\tau-\tau^2)-4\right).
\end{equation}
Let us thus focus now on the average winning probability $\overline{q}^u$.

\subsection{A confidence interval for the average winning probability}
The random variables $T_i^u$ in Eq.~\eqref{eq:Tiu} being estimators for the parameters $q_i^u$, we use their average
\begin{equation}\label{eq:Tu}
\overline{T}^u = \frac{1}{n} \sum_{i=1}^n T_i^u
\end{equation}
to estimate $\overline q^u$. This gives rise to the following effective CHSH value
\begin{equation}\label{eq:Su}
\overline{S}^u = 8\overline{T}^u-4
\end{equation}
which can be evaluated in practice directly from the observed data, without assumption on the distribution of measurement settings. Note that each random variable $T_i^u$ is a Bernoulli variable with parameter $q_i^u$. Therefore, $\overline{T}^u$ is a so-called (renormalized) Poisson binomial random variable. The distribution probability of such a random variable in terms of the average parameter $\overline q^u$ has been characterized by Hoeffding in 1956~\cite{Hoeffding56}. We recall this result here.

\begin{theorem}[Hoeffding, 1956]\label{thm:Hoeffding}
Let $\overline T=\frac{1}{n}\sum_{i=1}^n T_i$ be the average of $n$ independent Bernoulli variables $T_i$ with parameters $q_i$. If $c$ and $d$ are two integers such that
\begin{equation}
0 \leq c \leq n \overline q \leq d \leq n
\end{equation}
for $\overline q = \frac{1}{n}\sum_{i=1}^n q_i$, then
\begin{equation}
P(c\leq n \overline T \leq d) \geq \sum_{k=c}^d \binom{n}{k} {\overline q}^k (1-\overline q)^{n-k}.
\end{equation}
\end{theorem}

This theorem says that within all sets of $n$ choices of Bernoulli variables $\{T_i\}_{i=1}^n$ with a fixed average parameter $\overline q = \frac{1}{n}\sum_i q_i$, the one producing the largest tail distribution for the average variable $\overline T$ is the set of $n$ identically-distributed Bernoulli variables with $q_i=\overline q$ $\forall i$. The tail probability then follows a binomial distribution. Since $q_i=\overline q$ $\forall i$ is an admissible parameter value, this bound is tight.

We recall that the Binomial cumulative distribution can be expressed in terms of the regularized incomplete Beta function $I_\mx(\ma,\mb)$ as
\begin{equation}\label{eq:Isum1}
\sum_{k=0}^d \binom{n}{k} q^k (1-q)^{n-k} = 1-I_{q}(d+1,n-d),
\end{equation}
or equivalently
\begin{equation}\label{eq:Isum2}
\sum_{k=c}^n \binom{n}{k} q^k (1-q)^{n-k} = I_q(c,n-c+1).
\end{equation}
We then denote the inverse regularized incomplete Beta function by $I^{-1}$, i.e.~$I_\my(\ma,\mb)=\mx$ for $\my=I^{-1}_\mx(\ma,\mb)$.

Note that even if Thm.~\ref{thm:Hoeffding} applies to independent (though not necessarily identically distributed) random variables, it is still useful in our context, in which we do not wish to assume rounds to be either independent or identically distributed. The main reason for that is that our figure of merit is given by the average winning probability $\overline q^u$ defined as the average of individual winning probabilities $q_i^u=q^u_{i|\text{past}_i}$ \emph{conditioned} on past events. Hence, even if the physical process under study may depend on previous rounds, the random variables on which we would like to apply the theorem are statistically independent from each other: $P(T^u_i=t_i,T^u_{i+1}=t_{i+1})=P(T^u_i=t_i)P(T^u_{i+1}=t_{i+1})$. Indeed, the random variable $T^u_i$ of each round $i$ is fully characterized by a single and well-definied parameter $q^u_i$. We can thus use this result to 
construct a confidence interval for the average CHSH winning probability $\overline q^u$.

\begin{theorem}\label{thm:thm2}
Given Bernoulli random variables $T_i$ with parameter $q_i$ for $i=1,\ldots,n$ and $0 \leq \alpha \leq 1/2$, the interval $[\hat q,1]$ with the random variable
\begin{equation}
\hat{q} =
\begin{cases}
I^{-1}_{\alpha}(n \overline T-1, n(1-\overline T)+2) & \textrm{if } n\overline T \geq 1 \\
0 & \textrm{if } n\overline T = 0
\end{cases}
\end{equation}
is a confidence interval for $\overline q=\frac{1}{n}\sum_i q_i$ with confidence level $1-\alpha$.
\end{theorem}

\begin{proof}
We need to show that
\begin{equation}\label{eq:confInterval}
P(\hat q \leq \overline  q) \geq 1-\alpha
\end{equation}
for all possible sets of Bernoulli variables characterized by parameters $0\leq q_i\leq 1$ with $\frac{1}{n}\sum_i q_i = \overline q$. Condition~\eqref{eq:confInterval} states that whatever the unknown distribution of the Bernoulli variables $T_i$ happens to be, the value of $\hat q$ computed from them (the random variable $\hat q$ depends on the observed $\overline T$) can be higher than the actual parameter $\overline q$ only with probability at most $\alpha$. $\hat q$ then constitutes a reasonable lower bound for the parameter $\overline q$.

The case for which $n\overline T=0$ is clear. We can thus assume that $n\overline T\geq 1$. But before starting, let us introduce the function
\begin{equation}
\begin{split}
g:&\mathbb{R}\to\mathbb{R}\\
&z\mapsto I^{-1}_\alpha(z, n - z + 1).
\end{split}
\end{equation}
This function describes the trade-off between the parameters $q$ and $z$ in $I_q(z, n - z + 1)$. Let us remember that the incomplete beta function $I_q(\alpha,\beta)$ is the cumulative distribution function for the beta distribution with parameter $\alpha$ and $\beta$. Therefore, $I_q(z,n-z+1)$ is strictly increasing on $q\in[0,1]$ even for non-integer values of $n$ and $z$. At the same time, $1-I_q(z,n-z+1)$ seen as a function of $z\in\mathbb{R}$ is a cumulative distribution for the continuous binomial distribution with parameter $q$, so it strictly increases with $z\in[0,n+1]$~\cite{Ilienko13}. In other words, $I_q(z,n-z+1)$ strictly decreases with $z$. Since $g(z)$ is defined as the value of $q$ which leaves $I_q(z,n-z+1)$ invariant (and equal to $\alpha$) when $z$ changes, it is a strictly increasing function of $z\in[0,n+1]$: increasing $z$ increases $I_q(z,n-z+1)$ unless $q$ increases as well.

Let us now write
\begin{equation}\label{eq:P}
\begin{split}
P(\hat q \leq \overline q)
&=P(n\overline T = 0)\\
&\ \ \ + \sum_{k=1}^n P(n\overline T = k)\chi(g(k-1) \leq \overline q)\\
&= \sum_{k=0}^{d} P(n\overline T = k)\\
\end{split}
\end{equation}
where $P(n\overline T=k)$ is the probability distribution of the sum $n\overline T$ of arbitrary Bernoulli random variables and $d$ is the largest integer such that $g(d-1) \leq \overline q$, hence $g(d) > \overline q$. The sum contains all terms between 0 and $d$ because $g(z)$ is an increasing function.

Another implication of the monotonicity of $I_{\overline q}(z,n-z+1)$ with respect to q, is that its inverse function $I^{-1}_\alpha(z,n-z+1)$ increases with $\alpha$. Therefore
\begin{equation}
I^{-1}_\alpha(z,n-z+1) \leq I^{-1}_{1/2}(z,n-z+1)
\end{equation}
for $\alpha \leq 1/2$. $I^{-1}_{1/2}(z,n-z+1)$ is the median of a beta distribution, which can always be bounded by its mean~\cite{vanDeVen93}:
\begin{equation}
\begin{split}
I^{-1}_{1/2}(z,n-z+1) &\leq
\begin{cases}
\frac{z}{n+1} & \text{if }z \leq n-z+1\\
1-\frac{n-z+1}{n+1} & \text{else}
\end{cases}\\
&=\frac{z}{n+1}.
\end{split}
\end{equation}
Therefore, we have
\begin{equation}
I^{-1}_\alpha(n\overline q, n(1-\overline q)+1)\leq \frac{n\overline q}{n+1}\leq \overline q.
\end{equation}

Since $g(d) \geq \overline q$ by the definition of $d$, we obtain
\begin{align}
g(d) &\geq I^{-1}_\alpha(n\overline q, n(1-\overline q)+1)=g(n\overline q).
\end{align}
By the strict monotonicity of $g$, this implies that the condition $d\geq n\overline q$ is satisfied. So we can use Thm.~\eqref{thm:Hoeffding} to lower bound the probability Eq.~\eqref{eq:P} by the binomial case:
\begin{equation}\label{eq: P to binomial}
\begin{split}
P(\hat q \leq \overline q) &= \sum_{k=0}^{d} P(n\overline T = k) \\
&= P(n\overline T \leq d)\\
& \geq \sum_{k=0}^{d} \binom{n}{k} {\overline q}^k(1-\overline q)^{n-k}
\end{split}
\end{equation}

Using Eq.~\eqref{eq:Isum1} we obtain
\be\label{eq: P to beta}
P(\hat q \leq \overline q) \geq 1- I_{\overline q }(d+1, n-d).
\ee
Since $g(d)\geq \overline q$, or
\be\label{eq: Ialpha}
I_\alpha^{-1}(d, n-d+1) \geq \overline q
\ee
and $I_{\overline q}(d,n-d+1)$ is an increasing function of $\overline q$ we can write:
\begin{equation}
\begin{split}
I_{I_\alpha^{-1}(d,n-d+1)}(d,n-d+1)=\alpha &\geq I_{\overline q}(d,n-d+1)\\
&\geq I_{\overline q}(d+1,n-d),
\end{split}
\end{equation}
where we applied the incomplete beta function to both sides of Eq.~\eqref{eq: Ialpha} and used the monotonicity of the beta function again in the last line. Combining with Eq.~\eqref{eq: P to beta} completes the proof
\be
\begin{split}
P(\hat q \leq \overline q) &\geq 1- I_{\overline q }(d+1, n-d)\\
&\geq 1-\alpha.
\end{split}
\ee

\end{proof}

\section{Relation with the global fidelity}\label{app:Fidelities}
In this work, we quantify the quality of a multi-round state by its average Bell state fidelity over all rounds:
\begin{equation}\label{eq:ourFid}
\overline \cF = \frac{1}{n}\sum_{i=1}^n \cF_i,
\end{equation}
see Eq.~\eqref{eq:averageFidelity} of the main text. Other multi-round self-testing works have rather used the global fidelity~\cite{Supic19}
\begin{equation}\label{eq:otherFid}
F_g = \bra{\phi^+}^{\otimes n}\rho\ket{\phi^+}^{\otimes n},
\end{equation}
were $\ket{\phi^+}^{\otimes n}$ is the state of $n$ copies of an maximally entangled two-qubits state. The quantifier $F_g$ was used both in the sequential~\cite{Reichard13} and parallel repetition setting~\cite{McKague16,Natarajan17,Coladangelo17b}. In this appendix, we discuss the relation between these two fidelities. As we show below, strict bounds relate the two fidelity definitions, meaning that they are operationally equivalent (up to some rescaling and finite correction), see also~\cite{Sabat20}.

Although these two fidelities are equivalent to each other, it is worth noting that their actual values scale differently in presence of a finite level of experimental noise. To see this, consider the case of repeatedly measuring a Werner state
\begin{equation}\label{eq:WernerEpsilon}
\rho_i=(1-\frac{4}{3}\epsilon)\ketbra{\phi^+}{\phi^+} + \epsilon\frac{\Id}{3}
\end{equation}
at each round $i=1,\ldots,n$. In this case, the expression given in Eq.~\eqref{eq:otherFid} decreases with the number of rounds: $F_g=(1-\epsilon)^n\simeq 1-n\epsilon$ for small $\epsilon$. Therefore, this quantity does not directly reflect the quality of the setup that was used to create the state: knowledge of the number of rounds $n$ is needed to deduce the value of $\epsilon$ from $F_g$. In contrast, the value of Eq.~\eqref{eq:ourFid} is here $\overline \cF = 1-\epsilon$ for all $n$. Hence, the average fidelity does not depend on the number of rounds $n$ performed during the experiment and directly reflects the quality of the source.

For concreteness and simplicity, we now consider the estimation of both fidelities above on an arbitrary global state $\rho$ and with fixed extraction maps. This case is compatible with both the sequential and parallel repetition scenarios~\cite{Arnon-Friedman20}.

\begin{theorem}
When evaluating \eqref{eq:ourFid} and \eqref{eq:otherFid} an a global state $\rho\in L(\cH_A^{\otimes n} \otimes \cH_B^{\otimes n})$, the following inequalities hold:
\begin{equation}\label{eq:compareFids}
1-\overline\cF \leq 1-F_g \leq n(1-\overline\cF).
\end{equation}
Moreover, these bounds are tight.
\end{theorem}

\begin{proof}
Let us decompose the state $\rho$ across the $n$ rounds. For each round, we further decompose the Hilbert space of Alice and Bob into a first part spanned by $\ket{\phi^+}$ and the rest of the space. Since in the case of both fidelities we are only interested in the overlaps with the $\ket{\phi^+}$ state for each round, we can neglect all coherence between these subspaces. Without loss of generality, the full state can then be written in the form
\begin{equation}\label{eq:state1}
\rho = \sum_{\vec v\in \{0,1\}^{n}} c_{\vec v} \sigma_{\vec v}^1\otimes\sigma_{\vec{v}}^2\otimes\ldots.
\end{equation}
Here, $c_{\vec v}\geq 0$ $\forall \vec v$, $\sum_{\vec v\in \{0,1\}^n}c_{\vec v}=1$ and $\sigma_{\vec v}^i$ are quantum states s.t. $\sigma_{\vec v}^i=\ketbra{\phi^+}{\phi^+}$ for $v_i=0$ and $\tr{\sigma_{\vec v}^i\ketbra{\phi^+}{\phi^+}}=0$ otherwise.

Noting that the single round fidelities in~\eqref{eq:ourFid} are given by $\cF_i=\bra{\phi^+} \rho_i \ket{\phi^+}$, where $\rho_i$ is the partial trace of $\rho$ over all rounds except round $i$, the two fidelities can now be written explicitly in terms of the $c_{\vec v}$ coefficients of $\rho$:
\begin{eqnarray}
\cF_i &=& \sum_{\vec v:v_i=0} c_{\vec v}\\
F_g &=& c_{(0,0,\ldots,0)}.
\end{eqnarray}
Since the component $c_{(0,0,\ldots,0)}$ appears in both expressions, it is clear that $F_g$ cannot be larger than and $\cF_i$ for any $i$. Therefore, it also cannot be larger than their mean: $\overline\cF\geq F_g$, and we have the upper bound of Eq.~\eqref{eq:compareFids}. The choice $c_{1,1,\ldots,1}=1-c_{0,0,\ldots,0}$ saturates this bound.

To show the opposite bound, we first note that the quantities $\cF$ and $F_g$ are invariant under permutation of the rounds. We can thus symmetrize the state~\eqref{eq:state1} and express it in terms of just $n+1$ parameters $d_j$: after symmetrization, $\rho$ becomes
\begin{equation}
\rho' = d_n \ketbra{\phi^+}{\phi^+}^{\otimes n} + d_{n-1} \ketbra{\phi^+}{\phi^+}^{\otimes n-1}\otimes \tau + \ldots.
\end{equation}
The fidelities then take the form
\begin{eqnarray}\label{eq:fids}
\overline\cF=\cF_i &=& \sum_{j=1}^n \binom{n-1}{j-1} d_j\\
F_g &=& d_n
\end{eqnarray}
and we have the normalization condition
\begin{equation}\label{eq:norm}
\sum_{\vec v\in \{0,1\}^n}c_{\vec v}=\sum_{j=0}^n \binom{n}{j} d_j=1,
\end{equation}
and positivity condition $d_j\geq 0$.

We can now write
\begin{widetext}
\begin{eqnarray}
F_g &=& d_n \\
&=& \left((n-1)-(n-1)\right) d_0 + \sum_{j=1}^{n-1}\left[\left((n-1)\binom{n}{j} - n\binom{n-1}{j-1}\right) - (n-1)\binom{n}{i} + n\binom{n-1}{j-1}\right] d_j + d_n\nonumber\\
&=& (n-1) d_0 + \sum_{j=1}^{n-1}\left((n-1)\binom{n}{j} - n\binom{n-1}{j-1}\right)d_j - (n-1)\sum_{j=0}^n\binom{n}{j}d_j + n\sum_{j=1}^n\binom{n-1}{j-1}d_j\\
&\geq& 1 - n(1-\overline \cF)
\end{eqnarray}
\end{widetext}
as desired. Here, we used the relations above and $(n-1)\binom{n}{j} - n\binom{n-1}{j-1} \geq 0$, which is true for $j\leq n-1$. The inequality $1-F_g\geq n(1-\overline\cF)$ is saturated for the choice $d_{n-1}=(1-d_n)/n$, $d_j=0\ \forall j\leq n-2$.
\end{proof}

\bibliographystyle{apsrev4-1_modified}
\nocite{apsrev41Control} 
\bibliography{references.bib}

\end{document}